\pdfoutput=1 

\documentclass[runningheads]{llncs}
\usepackage{graphicx}
\usepackage{color}
\usepackage{hyperref}
\usepackage{float}

\usepackage[ruled,vlined,linesnumbered]{algorithm2e}
\usepackage{enumitem}

\usepackage{amsmath,amssymb}
\usepackage{tikz}
\usepackage{forest}
\usepackage{comment}

\begin{document}
\title{Optimal monomial quadratization\\ for ODE systems\thanks{The article was prepared within the framework of the HSE University Basic Research Program. GP was partially supported by NSF grants DMS-1853482, DMS-1760448, DMS-1853650, CCF-1564132, and CCF-1563942 and by the Paris Ile-de-France region. The authors are grateful to Mathieu Hemery, Fran\c{c}ois Fages, and Sylvain Soliman for helpful discussions.
The work has started when G. Pogudin worked at the Higher School of Economics, Moscow.
The authors would like to thank the referees for their comments, which helped us improve the manuscript.}}

\author{Andrey Bychkov\inst{1}
\and
Gleb Pogudin\inst{2}
}
\authorrunning{A. Bychkov, G. Pogudin}

\institute{
Higher School of Economics, Myasnitskaya str., 101978 Moscow, Russia
\email{abychkov@edu.hse.ru} \and
LIX, CNRS, \'Ecole Polytechnique, Institute Polytechnique de Paris, Palaiseau, France \email{gleb.pogudin@polytechnique.edu}
}

\maketitle              
\begin{abstract}
Quadratization is a transform of a system of ODEs with polynomial right-hand side into a system of ODEs with at most quadratic right-hand side via the introduction of new variables.
Quadratization problem is, given a system of ODEs with polynomial right-hand side, transform the system to a system with quadratic right-hand side by introducing new variables. 
Such transformations have been used, for example, as a preprocessing step by model order reduction methods and for transforming chemical reaction networks.

We present an algorithm that, given a system of polynomial ODEs, finds a transformation into a quadratic ODE system by introducing new variables which are monomials in the original variables.
The algorithm is guaranteed to produce an optimal transformation of this form (that is, the number of new variables is as small as possible), and it is the first algorithm with such a guarantee we are aware of.
Its performance compares favorably with the existing software, and it is capable to tackle problems that were out of reach before.

\keywords{differential equations  \and branch-and-bound \and quadratization.}
\end{abstract}
\section{Introduction}

The \emph{quadratization} problem considered in this paper is, given a system of ordinary differential equations (ODEs) with polynomial right-hand side, transform it into a system with quadratic right-hand side (see Definition~\ref{def:quadr}).
We illustrate the problem on a simple example of a scalar ODE:
\begin{equation}\label{eq:ex_start}
x' = x^5.
\end{equation}
The right-hand side has degree larger than two but if we introduce a new variable $y := x^4$, then we can write:
\begin{equation}\label{eq:ex_finish}
x' = xy,\quad \text{ and }\quad y' = 4x^3 x' = 4x^4y = 4y^2.
\end{equation}
The right-hand sides of~\eqref{eq:ex_finish} are of degree at most two, and every solution of~\eqref{eq:ex_start} is the $x$-component of some solution of~\eqref{eq:ex_finish}.

A problem of finding such a transformation (\emph{quadratization}) for an ODE system has appeared recently in several contexts:
\begin{itemize}
    \item One of the recent approaches to \emph{model order reduction}~\cite{Gu} uses quadratization as follows.
    For the ODE systems with quadratic right-hand side, there are dedicated model order reduction methods which can produce a better reduction than the general ones.
    Therefore, it can be beneficial to perform a quadratization first and then use the dedicated methods.
    For further details and examples of applications, we refer to~\cite{Gu,KW19b,KW19a,RWBG20}.
    \item Quadratization has been used as a preprocessing step for \emph{solving differential equations numerically}~\cite{CV1,GCV2019,KCV2013}.
    \item Applied to~\emph{chemical reaction networks}, quadratization allows one to transform a given chemical reaction network into a bimolecular one~\cite{CRN}.
\end{itemize}

It is known (e.g.~\cite[Theorem~3]{Gu}) that it is always possible to perform quadratization with new variables being monomials in the original variables (like $x^4$ in the example above).
We will call such quadratization \emph{monomial} (see Definition~\ref{def:monomial_quadr}).
An algorithm for finding some monomial quadratization has been described in~\cite[Section G.]{Gu}.
In~\cite{CRN}, the authors have shown that the problem of finding an optimal (i.e. of the smallest possible dimension) monomial quadratization is NP-hard.
They also designed and implemented an algorithm for finding a monomial quadratization which is practical and yields an optimal monomial quadratization in many cases (but not always, see Section~\ref{sec:prior}).

In this paper, we present an algorithm that computes an optimal monomial quadratization for a given system of ODEs.
To the best of our knowledge, this is the first practical algorithm with the optimality guarantee.
In terms of efficiency, our implementation compares favorably to the existing software~\cite{CRN} (see Table~\ref{table:CRM_comp}).
The implementation is publicly available at~\url{https://github.com/AndreyBychkov/QBee/}.
Our algorithm follows the classical Branch-and-Bound approach~\cite{BBsurvey} together with problem-specific search and branching strategies and pruning rules (with one using the extremal graph theory, see Section~\ref{sec:ruleC4}).

Note that, according to~\cite{Foyez}, one may be able to find a quadratization of lower dimension by allowing the new variables to be arbitrary polynomials, not just monomials.
We restrict ourselves to the monomial case because it is already challenging (e.g., includes an APX-hard [2]-sumset cover problem, see Remark~\ref{rem:2sumsetcover}) and monomial transformations are relevant for some application areas~\cite{CRN}.

The rest of the paper is organized as follows.
In Section~\ref{sec:statement}, we state the problem precisely.
In Section~\ref{sec:prior}, we review the prior approaches, most notably~\cite{CRN}.
Sections~\ref{sec:alg} and~\ref{sect:pruning_rules} describe our algorithm.
Its performance is demonstrated and compared to~\cite{CRN} in Section~\ref{sec:performace}.
Sections~\ref{sec:complexity} and~\ref{sec:conclusion} contain remarks on the complexity and conclusions/open problems, respectively.


\section{Problem Statement}\label{sec:statement}

\begin{definition}\label{def:quadr}
  Consider a system of ODEs
  \begin{equation}\label{eq:sys_main}
  x_1' = f_1(\bar{x}),\;\;\ldots,\;\;\ x_n' = f_n(\bar{x}),
  \end{equation}
  where $\bar{x} = (x_1, \ldots, x_n)$ and 
  $f_1, \ldots, f_n \in \mathbb{C}[\mathbf{x}]$.
  Then a list of new variables
  \begin{equation}\label{eq:quadr}
      y_1 = g_1(\bar{x}), \ldots, y_m = g_m(\bar{x}),
  \end{equation}
  is said to be a \emph{quadratization} of~\eqref{eq:sys_main} if there exist polynomials $h_1, \ldots, h_{m + n} \in \mathbb{C}[\bar{x}, \bar{y}]$ of degree at most two such that 
  \begin{itemize}
      \item $x_i' = h_i(\bar{x}, \bar{y})$ for every $1 \leqslant i \leqslant n$;
      \item $y_j' = h_{j + n}(\bar{x}, \bar{y})$ for every $1 \leqslant j \leqslant m$.
  \end{itemize}
  The number $m$ will be called \emph{the order of quadratization}.
  A quadratization of the smallest possible order will be called \emph{an optimal quadratization}.
\end{definition}

\begin{definition}\label{def:monomial_quadr}
  If all the polynomials $g_1, \ldots, g_m$ are monomials, the quadratization is called \emph{a monomial quadratization}.
  If a monomial quadratization of a system has the smallest possible order among all the monomial quadratizations of the system, it is called \emph{an optimal monomial quadratization}.
\end{definition}

Now we are ready to precisely state the main problem we tackle:
\begin{description}
  \item[Input] A system of ODEs of the form~\eqref{eq:sys_main}.
  \item[Output] An optimal monomial quadratization of the system.
\end{description}

\begin{example} \label{ex:simp_example}
  Consider a single scalar ODE $x' = x^5$ from~\eqref{eq:ex_start}, that is $f_1(x) = x^5$.
  As has been show in~\eqref{eq:ex_finish}, $y = x^4$ is a quadratization of the ODE with $g(x) = x^4$, $h_1(x, y) = xy$, and $h_2(x, y) = 4y^2$.
  Moreover, this is a monomial quadratization.
  
  Since the original ODE is not quadratic, the quadratization is optimal, so it is also an optimal monomial quadratization.
\end{example}

\begin{example}
\textit{The Rabinovich-Fabrikant system}~\cite[Eq. (2)]{RF1} is defined as follows:
\[
  x' = y (z - 1 + x^2) + ax,\;\;  y' = x(3z + 1 - x^2) + ay,\;\; z' = -2z (b + xy).
\]
Our algorithm finds an optimal monomial quadratization of order three: $z_1 = x^2, z_2 = xy, z_3 = y^2$.
The resulting quadratic system is:
\begin{align*}
		x' &= y(z_{1} + z - 1) + ax,& z_{1}' &= 2z_{1} (a + z_{2}) + 2 z_{2}( z - 1),\\
		y' &= x(3z + 1 - z_{1}) + ay,& z_{2}' &= 2 a z_{2} + z_{1}(3z + 1 - z_{1} +  z_{3}) + z_{3}(z - 1)\\
		z' &= -2z(b + z_{2}),& z_{3}' &= 2 a z_{3} + 2z_{2}(3z + 1 - z_{1}).
\end{align*}  
\end{example}


\section{Discussion of prior approaches}\label{sec:prior}

To the best of our knowledge, the existing algorithms for quadratization are~\cite[Algotirhm~2]{Gu} and~\cite[Algorithm~2]{CRN}.
The former has not been implemented and is not aimed at producing an optimal quadratization: it simply adds new variables until the system is quadratized, and its termination is based on~\cite[Theorem~2]{Gu}.

It has been shown~\cite[Theorem~2]{CRN} that finding an optimal quadratization is NP-hard.
The authors designed and implemented an algorithm for finding a small (but not necessarily optimal) monomial quadratization which proceeds as follows.
For an $n$-dimensional 
system $\bar{x}' = \bar{f}(\bar{x})$, define, for every $1 \leqslant i \leqslant n$,
\[
    D_i := \max\limits_{1 \leqslant j \leqslant n} \deg_{x_i} f_j.
\]
Then consider the set 
\begin{equation}\label{eq:setM}
M := \{ x_1^{d_1} \ldots x_n^{d_n} \mid 0 \leqslant d_1 \leqslant D_1, \ldots, 0 \leqslant d_n \leqslant D_n\}.
\end{equation}
\cite[Proof of Theorem~1]{CPSW05} implies that there always exists a monomial quadratization with the new variables from $M$.
The idea behind~\cite[Algorithm~2]{CRN} is to search for an optimal quadratization \emph{inside} $M$. 
This is done by an elegant encoding into a MAX-SAT problem.

However, it turns out that the set $M$ does not necessarily contain an optimal monomial quadratization.
As our algorithm shows, this happens, for example, for some of the benchmark problems from~\cite{CRN} (Hard and Monom series, see~Table~\ref{table:CRM_comp}).
Below we show a simpler example illustrating this phenomenon.

\begin{example}
  Consider a system
  \begin{equation}\label{eq:counterexample}
    x_1' = x_2^4, \quad x_2' = x_1^2.
  \end{equation}
  Our algorithm shows that it has a unique optimal monomial quadratization
  \begin{equation}\label{eq:new_vars_counterex}
    z_1 = x_1 x_2^2,\;\; z_2 = x_2^3,\;\; z_3 = x_1^3
  \end{equation}
  yielding the following quadratic ODE system:
  \begin{align*}
      x_1' &= x_2z_2, & z_1' &= x_2^6 + 2x_1^3x_2 = z_2^2 + 2x_2z_3, & z_3' &= 3x_1^2x_2^4 = 3z_1^2,\\
      x_2' &= x_1^2, & z_2' &= 3x_1^2x_2^2 = 3x_1z_1. & & 
  \end{align*}
  The degree of~\eqref{eq:new_vars_counterex} with respect to $x_1$ is larger than the $x_1$-degree of the original system~\eqref{eq:counterexample}, so such a quadratization will not be found by the algorithm~\cite{CRN}.
\end{example}

It would be interesting to find an analogue of the set $M$ from~\eqref{eq:setM} always containing an optimal monomial quadratization as this would allow using powerful SAT-solvers.
For all the examples we have considered, the following set worked
\[
  \widetilde{M} := \{ x_1^{d_1} \ldots x_n^{d_n} \mid 0 \leqslant d_1 , \ldots,  d_n \leqslant D\}, \quad \text{where } D := \max\limits_{1 \leqslant i \leqslant n} D_i.
\]


\section{Outline of the algorithm}\label{sec:alg}

Our algorithm follows the general  \emph{Branch-and-Bound (B\&B)} paradigm~\cite{BBsurvey}.
We will describe our algorithm using the standard B\&B terminology (see, e.g., \cite[Section~2.1]{BBsurvey}).

\begin{definition}[B\&B formulation for the quadratization problem]\label{def:bb}
  \begin{itemize}
      \item \emph{The search space} is a set of all monomial quadratizations of the input system $\bar{x}' = \bar{f}(\bar{x})$.
      \item \emph{The objective function} to be minimized is the number of new variables introduced by a quadratization.
      \item Each \emph{subproblem} is defined by a set of new monomial variables $z_1(\bar{x}), \ldots, z_\ell(\bar{x})$ and the corresponding subset of the search space is the set of all quadratizations including the variables $z_1(\bar{x}), \ldots, z_\ell(\bar{x})$.
  \end{itemize}
\end{definition}

\begin{definition}[Properties of a subproblem]\label{def:vars_nonsq}
To each subproblem (see Definition~\ref{def:bb}) defined by new variables $z_1(\bar{x}), \ldots, z_\ell(\bar{x})$, we assign:
\begin{enumerate}
    \item the set of \emph{generalized variables}, denoted by $V$, consisting of the polynomials $1, x_1, \ldots, x_n, z_1(\bar{x}), \ldots, z_\ell(\bar{x})$;
    \item the set of \emph{nonsquares}, denoted by $\operatorname{NS}$, consisting of all the monomials in the derivatives of the generalized variables which do not belong to $V^2 := \{v_1v_2 \mid v_1, v_2 \in V\}$. 
    In particular, a subproblem is a quadratization iff $\operatorname{NS} = \varnothing$.
\end{enumerate}
\end{definition}

\begin{example}\label{ex:V_NS}
  We will illustrate the notation introduced in Definition~\ref{def:vars_nonsq} on a system $x' = x^4 + x^3$ and a new variable $z_1(x) = x^3$.
  We have $z_1' = 3x^2 x' = 3x^6 + 3x^5$.
  Therefore, for this subproblem, we have:
  \[
    V = \{ 1, x, x^3 \}, \quad V^2 = \{1, x, x^2, x^3, x^4, x^6\}, \quad \operatorname{NS} = \{x^5\}.
  \]
\end{example}

In order to organize a B\&B search in the search space defined above, we define several subroutines/strategies answering the following questions:
\begin{itemize}
    \item \emph{How to set the original bound?}
    \cite[Theorem~1]{CPSW05} implies that the set $M$ from~\eqref{eq:setM} gives a quadratization of the original system, so it can be used as the starting incumbent solution.
    \item \emph{How to explore the search space?}
    There are two subquestions:
    \begin{itemize}
        \item \emph{What are the child subproblems of a given subproblem (branching strategy)?}
        This is described in Section~\ref{sec:branching}.
        \item \emph{In what order we traverse the tree of the subproblems?}
        We use DFS (to make new incumbents appear earlier) guided by a heuristic as described in Algorithm~\ref{alg:BnB_main}.
    \end{itemize}
    \item \emph{How to prune the search tree (pruning strategy)?}
    We use two algorithms for computing a lower bound for the objective function in a given subtree, they are described and justified in Section~\ref{sect:pruning_rules}.
\end{itemize}

\subsection{Branching strategy}\label{sec:branching}

Let $\bar{x}' = \bar{f}(\bar{x})$ be the input system.
Consider a subproblem defined by new monomial variables $z_1(\bar{x}), \ldots, z_\ell(\bar{x})$.
The child subproblems will be constructed as follows:
\begin{enumerate}
    \item among the nonsquares ($\operatorname{NS}$, see Definition~\ref{def:vars_nonsq}), choose any monomial $m = x_1^{d_1}\ldots x_n^{d_n}$ with the value $\prod_{i = 1}^n (d_i + 1)$ the smallest possible;
    \item for every decomposition $m = m_1 m_2$ as a product of two monomials, define a new subproblem by adding the elements of $\{m_1, m_2\} \setminus V$ (see Definition~\ref{def:vars_nonsq}) as new variables.
    Since $m \in \operatorname{NS}$, at least one new variable will be added.
\end{enumerate}

The score function $\prod_{i = 1}^n (d_i + 1)$ is twice the number of representations $m = m_1m_2$, so this way we reduce the branching factor of the algorithm.

\begin{lemma}\label{lem:branching}
    Any optimal subproblem $z_1(\bar{x}), \ldots, z_\ell(\bar{x})$ is a solution of at least one of the children subproblems generated by the procedure above.
\end{lemma}

\begin{proof}
   Let $z_1(\bar{x}), \ldots, z_n(\bar{x})$ be any solution of the subproblem.
   Since $m$ must be either of the form $z_i z_j$ or $z_j$, it will be a solution of the child subproblem corresponding to the decomposition $m = z_i z_j$ or $m = 1 \cdot z_j$, respectively.
\end{proof}

\begin{example}
    Figure~\ref{fig:simple_graph} below show the graph representation of system $x' = x^4 + x^3$ from Example~\ref{ex:V_NS}.
    The starting vertex is $\varnothing$.
    The underlined vertices correspond to optimal quadratizations, so the algorithm will return one of them. 
    On the first step, the algorithm chooses the monomial $x^3$ which has two decompositions $x^3 = x \cdot x^2$ and $x^3 = 1 \cdot x^3$ yielding the left and the right children of the root, respectively.
    The subproblem $\{x^3\}$ was described in more details in Example~\ref{ex:V_NS}.
    
    The score function $\prod_{i = 1}^n(d_i + 1)$ for the decompositions $x^3 = x \cdot x^2$ and $x^3 = 1 \cdot x^3$ takes values $6$ and $4$, respectively.
    Hence the algorithm will first explore the branch on the right.
    \begin{figure}
    \centering
    \begin{tikzpicture}
        \node (0) at ( 0, 0) {$\varnothing$}; 
        \node (x2) at ( -2, -1) {$\{x^2\}$};
        \node (x3) at ( 2, -1) {$\{x^3\}$};
        \node (x2x3) at ( 0, -2) {$\underline{\{x^2,x^3\}}$};
        \node (x2x4) at ( -4,-2) {$\{x^2,x^4\}$};
        \node (x2x5) at ( -2,-2) {$\{x^2,x^5\}$};
        \node (x3x4) at ( 2,-2) {$\underline{\{x^3,x^4\}}$};
        \node (x3x5) at ( 4,-2) {$\{x^3,x^5\}$};
    
        \begin{scope}[every path/.style={->}]
           \draw (0) -- (x2);
           \draw (0) -- (x3);
           \draw (x2) -- (x2x3);
           \draw (x2) -- (x2x4);
           \draw (x2) -- (x2x5);
           \draw (x3) -- (x2x3);
           \draw (x3) -- (x3x4);
           \draw (x3) -- (x3x5);
        \end{scope}  
    \end{tikzpicture}

    \caption{Graph illustration for equation $x' = x^4 + x^3$}
    \label{fig:simple_graph}
\end{figure}
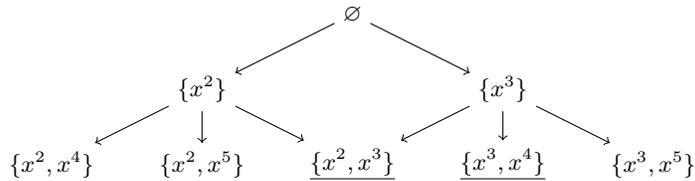
\end{example}

\subsection{Recursive step of the algorithm}
The recursive step of our algorithm can be described as follows.

{\small
\begin{algorithm}[H]
\caption{Branch and Bound recursive step}\label{alg:BnB_main}

\begin{description}
  \item[Input]
  \begin{itemize}
      \item[]
      \item polynomial ODE system $\bar{x}' = \bar{f}(\bar{x})$;
      \item set of new variables $z_1(\bar{x}), \ldots, z_\ell(\bar{x})$;
      \item an optimal quadratization found so far (incumbent) with $N$ new variables.
  \end{itemize}
  \item[Output] the algorithm replaces the incumbent with a more optimal quadratization containing $z_1(\bar{x}), \ldots, z_\ell(\bar{x})$ if such quadratization exists.
\end{description}

\begin{enumerate}[label = \textbf{(Step~\arabic*)}, leftmargin=*, align=left, labelsep=2pt, itemsep=0pt]
    \item if $z_1(\bar{x}), \ldots, z_\ell(\bar{x})$ is a quadratization
    \begin{enumerate}
        \item if $\ell < N$, replace the incumbent with $z_1(\bar{x}), \ldots, z_\ell(\bar{x})$;
        \item \KwRet;
    \end{enumerate}
    \item if any of the pruning rules (Algorithm~\ref{alg:early_quadratic} or~\ref{alg:early_C4}) applied to $z_1(\bar{x}), \ldots, z_\ell(\bar{x})$ and $N$ return \texttt{True}, \KwRet;
    \item generate set $C$ of child subproblems as described in Section~\ref{sec:branching}
    \item sort $C$ in increasing order w.r.t. $S + n |V|$, where $S$ is the sum of the degrees of the elements in $V$ ($V$ is different for different subproblems as defined in Definition~\ref{def:vars_nonsq});
    \item for each element of $C$, call Algorithm~\ref{alg:BnB_main} on it.
\end{enumerate}
\end{algorithm}}
 
\section{Pruning rules}\label{sect:pruning_rules}

In this section, we present two pruning rules yielding a substantial speedup of the algorithm: based on a quadratic upper bound and based on squarefree graphs.

\begin{property}\label{prop:early_io_spec}
Each pruning rule has the following input-output specification:
\begin{description}

    \item[Input:] \begin{itemize}
        \item[]
        \item the original ODE system $\bar{x}' = \bar{f}(\bar{x})$;
        \item already added new variables $z_1(\bar{x}), \ldots, z_\ell(\bar{x})$ which are monomials in $\bar{x}$;
        \item positive integer $N$.
    \end{itemize}
    
    \item[Output:] \texttt{True} if it is guaranteed that the set of new variables $z_1(\bar{x}), \ldots, z_s(\bar{x})$ cannot be extended to a monomial quadratization of $\bar{x}' = \bar{f}(\bar{x})$ of order less then $N$.
    \texttt{False} otherwise.
\end{description}
\end{property}

Note that, if \texttt{False} is returned, it does not imply that the set of new variables can be extended.

\begin{remark}
   Both pruning rules presented here actually check a stronger condition: whether the set of new variables can be extended by at most $N - s$ variables so that all the monomials $\operatorname{NS}$ in the \emph{current} subproblem can be written as a product of two generalized variables.
   It would be very interesting to strengthen these rules by taking into account the derivatives of the extra new variables.
\end{remark}

\subsection{Rule based on quadratic upper bound}

\begin{remark}[Intuition behind the rule]\label{rem:quadr_rationale}
Consider a subproblem with the generalized variables $V$ and set of nonsquares $\operatorname{NS}$ (see Definition~\ref{def:vars_nonsq}).
Assume that it can be quadratized by adding a set $W$ of variables.
This would imply that $\operatorname{NS} \subseteq (V \cup W)^2$.
This yields a bound 
\begin{equation}\label{eq:quadratic_rationale}
  |\operatorname{NS}| \leqslant \frac{(|V| + |W|)(|V| + |W| + 1)}{2}.
\end{equation}
The general ideal of the rule is: since $|V|$ and $|\operatorname{NS}|$ are known, \eqref{eq:quadratic_rationale} can be used to find a lower bound for $|W|$.
However, a straightforward application of~\eqref{eq:quadratic_rationale} does not lead to noticeable performance improvements.
We found that one can do much better by first estimating the number of elements of $\operatorname{NS} \cap (V \cdot W)$ and then applying an argument as in~\eqref{eq:quadratic_rationale} to $\operatorname{NS} \setminus (V \cdot W)$ and $W$.
\end{remark}

{\small
\begin{algorithm}[H]
\caption{Pruning rule: based on a quadratic upper bound}\label{alg:early_quadratic}

\begin{enumerate}[label = \textbf{(Step~\arabic*)}, leftmargin=*, align=left, labelsep=2pt, itemsep=0pt]
    \item\label{step:setD} Compute the following multiset of monomials in $\bar{x}$
    \[
      D := \{ m / v  \mid m \in \operatorname{NS}, v \in V, v \mid m\}.
    \]
    \item Let \texttt{mult} be the list of multiplicities of the elements of $D$ sorted in the descending order.
    \item Find the smallest integer $k$ such that 
    \begin{equation}\label{eq:square_bound}
      |\operatorname{NS}| \leqslant \sum\limits_{i = 1}^k \operatorname{mult}[i] + \frac{k(k + 1)}{2}.  
    \end{equation}
    (We use 1-based indexing and set $\operatorname{mult}[i] = 0$ for $i > |\operatorname{mult}|$)
    \item If $k + \ell \geqslant N$, return \texttt{True}. Otherwise, return \texttt{False}.
\end{enumerate}
\end{algorithm}}

\begin{lemma}\label{lem:square_bound}
    Algorithm~\ref{alg:early_quadratic} satisfied the specification described in Property~\ref{prop:early_io_spec}.
\end{lemma}

\begin{proof}
   Assume that Algorithm~\ref{alg:early_quadratic} has returned \texttt{True}.
   Consider any quadratization $z_1, \ldots, z_{\ell + r}$ of $\bar{x}' = \bar{f}(\bar{x})$ extending $z_1, \ldots, z_\ell$.
   We define $\widetilde{V}$, a superset of $V$, as $\{1, x_1, \ldots, x_n, z_1, \ldots, z_{\ell + r}\}$.
   By the definition of quadratization, $\operatorname{NS} \subseteq \widetilde{V}^2$.
   We split $\operatorname{NS}$ into two subsets $\operatorname{NS}_0 := \operatorname{NS} \cap (V \cdot \widetilde{V})$ and $\operatorname{NS}_1 := \operatorname{NS} \setminus \operatorname{NS}_0$.
   For every $1 \leqslant i \leqslant r$, the cardinality of $\operatorname{NS} \cap (z_{\ell + i} \cdot V)$ does not exceed the multiplicity of $z_{\ell + i}$ in the multiset $D$ constructed at~\ref{step:setD}.
   Therefore, $|\operatorname{NS}_0| \leqslant \sum\limits_{i = 1}^r \operatorname{mult}[i]$.
   The number of products of the form $z_{\ell + i}z_{\ell + j}$ with $1 \leqslant i \leqslant j \leqslant r$ does not exceed $\frac{r(r + 1)}{2}$.
   Therefore, we have
   \[
     |\operatorname{NS}| = |\operatorname{NS}_0| + |\operatorname{NS}_1| \leqslant \sum\limits_{i = 1}^r \operatorname{mult}[i] + \frac{r(r + 1)}{2},
   \]
   so $r$ satisfies~\eqref{eq:square_bound}.
   The minimality of $k$ implies $r \geqslant k$.
   Thus, $r + \ell \geqslant N$, so $z_1, \ldots, z_\ell$ cannot be extended to a quadratization of order less than $N$.
\end{proof}


\subsection{Rule based on squarefree graphs}\label{sec:ruleC4}

\begin{remark}[Intuition behind the rule]
We will illustrate the idea behind the rule in a simple example.
Assume that we have five monomials $m_1, \ldots, m_5$ such that none of them is a square.
Assume also that there is a set $V$ of monomial new variables such that $|V| = 4$ and $m_i \in V^2$ for every $i$.
Since none of $m_i$'s is a square, it can be written as $m_i = z_{i, 1} z_{i, 2}$ for distinct $z_{i, 1}, z_{i, 2} \in V$.
We can therefore think about a graph with vertices being elements of $V$ and edges given by $m_1, \ldots, m_5$.
One can check that every graph with four vertices and five edges must contain a four-cycle.
Let the cycle consist of edges $m_1, m_2, m_3, m_4$ in this order.
Then, for some numbering of elements in $V$, we have:
\[
  m_1 = z_1 z_2,\; m_2 = z_2 z_3,\; m_3 = z_3 z_4,\; m_4 = z_4 z_1 \implies m_1 m_3 = m_2 m_4.
\]
Thus, by checking that all pairwise product of $m_1, \ldots, m_5$ are distinct, we can verify that $m_1, \ldots, m_5 \in V^2$ implies that $|V| > 4$.

In order to take into account the monomials which are squares, we consider not just graphs but pseudographs.
We also employ the separation strategy $\operatorname{NS} = (\operatorname{NS} \cap (V\cdot W)) \cup (\operatorname{NS} \setminus (V\cdot W))$ as described in Remark~\ref{rem:quadr_rationale}.
\end{remark}

\begin{definition}
  A pseudograph $G$ (i.e., a graph with loops and multiple edges allowed) is called \emph{$C4^\ast$-free} if there is no cycle of length four in $G$ with every two adjacent edges being distinct (repetition of edges and/or vertices is allowed).
\end{definition}

\begin{example}
  A $C4^\ast$-free pseudograph cannot contain:
  \begin{itemize}[topsep=0pt]
      \item \emph{A vertex with two loops.}
      If the loops are $\ell_1$ and $\ell_2$ then the cycle $\ell_1, \ell_2, \ell_1, \ell_2$ will violate $C4^\ast$-freeness.
      \item \emph{Multiple edges.} If $e_1$ and $e_2$ are edges with the same endpoints, then $e_1, e_2, e_1, e_2$ will violate $C4^\ast$-freeness.
      \item \emph{Two vertices with loops connected by an edge.} If the loops are $\ell_1$ and $\ell_2$ and the edge is $e$, then $\ell_1, e, \ell_2, e$ will violate $C4^\ast$-freeness.
  \end{itemize}
\end{example}

\begin{definition}\label{def:Cnm}
  By $C(n, m)$ we denote the largest possible number of edges in a $C4^\ast$-free pseudograph $G$ with $n$ vertices and at most $m$ loops.
\end{definition}

\begin{remark}
  Note that the example above implies that $C(n, n + k) = C(n, n)$ for every positive integer $k$ because a $C4^\ast$-free pseudograph cannot contain more than $n$ loops.
  
  The number $C(n, 0)$ is the maximal number of edges in a $C4$-free graph and has been extensively studied (e.g.~\cite{C42,Clapham1989,C4ERS,C4}). 
  Values for $n \leqslant 31$ are available as a sequence A006855 in OEIS~\cite{OEIS}.
 \end{remark}
 
In Algorithm~\ref{alg:early_C4}, we use the exact values for $C(n, m)$ found by an exhaustive search and collected in Table~\ref{table:Cnm} for $n \leqslant 7$. The script for the search is available at~\url{https://github.com/AndreyBychkov/QBee/blob/0.5.0/qbee/no_C4_count.py}.
For $n > 7$, we use the following bound
\[
  C(n, m) \leqslant C(n, 0) + m \leqslant \frac{n}{2}(1 + \sqrt{4n - 3}) + m,
\]
where the bound for $C(n, 0)$ is due to~\cite[Chapter 23, Theorem~1.3.3]{HC}.

\begin{table}[h]
\centering
\setlength{\tabcolsep}{0.5em} 
{\renewcommand{\arraystretch}{1.2}
\begin{tabular}{c|*{8}{c|}}
 & \multicolumn{8}{l}{m}\\
n & 0 & 1 & 2 & 3 & 4 & 5 & 6 & 7 \\
\hline
1 & 0 & 1 &   &   &   &   &   &   \\ 
2 & 1 & 2 & 2 &   &   &   &   &   \\ 
3 & 3 & 3 & 4 & 4 &   &   &   &   \\ 
4 & 4 & 5 & 5 & 6 & 6 &   &   &   \\ 
5 & 6 & 6 & 7 & 7 & 8 & 8 &   &   \\ 
6 & 7 & 8 & 9 & 9 & 9 & 10 & 10 &   \\
7 & 9 & 10 & 11 & 12 & 12 & 12 & 12 &  12 \\
\hline
\end{tabular}
\vspace{2mm}
\caption{Exact values for $C(n, m)$ (see Definition~\ref{def:Cnm}).}\label{table:Cnm}
}
\end{table}

{\small
\begin{algorithm}[H]
\caption{Pruning rule: based on squarefree graphs}\label{alg:early_C4}

\begin{enumerate}[label = \textbf{(Step~\arabic*)}, leftmargin=*, align=left, labelsep=2pt, itemsep=0pt]
    \item\label{step:constructE} Compute a subset $E = \{m_1, \ldots, m_e\} \subseteq \operatorname{NS}$ such that all the products $m_im_j$ for $1 \leqslant i \leqslant j \leqslant e$ are distinct.
    
    (done by traversing $\operatorname{NS}$ in a descending order w.r.t. the total degree and appending each monomial if it does not violate the property)
    
    \item\label{step:setD_C4} Compute the following multiset of monomials in $\mathbf{x}$
    \[
      D := \{ m / v  \mid m \in E, v \in V, v \mid m\}.
    \]
    
    \item Let \texttt{mult} be the list of multiplicities of the elements of $D$ sorted in descending order.
    
    \item Let $c$ be the number of elements in $E$ with all the degrees being even.
    
    \item Find the smallest integer $k$ such that 
    \begin{equation}\label{eq:C4_bound}
      |E| \leqslant \sum\limits_{i = 1}^k \operatorname{mult}[i] + C(k, c).  
    \end{equation}
    (We use 1-based indexing and set $\operatorname{mult}[i] = 0$ for $i > |\operatorname{mult}|$)
    \item If $k + \ell \geqslant N$, return \texttt{True}. Otherwise, return \texttt{False}.
\end{enumerate}
\end{algorithm}}

\begin{lemma}
    Algorithm~\ref{alg:early_C4} satisfied the specification described in Property~\ref{prop:early_io_spec}.
\end{lemma}

\begin{proof}
   Assume that Algorithm~\ref{alg:early_quadratic} has returned \texttt{True}.
   Consider any quadratization $z_1, \ldots, z_{\ell + r}$ of $\bar{x}' = \bar{f}(\bar{x})$ extending $z_1, \ldots, z_\ell$.
   We define $\widetilde{V}$, a superset of $V$, as $\{1, x_1, \ldots, x_n, z_1, \ldots, z_{\ell + r}\}$.
   By the definition of quadratization, $E \subseteq \operatorname{NS} \subseteq \widetilde{V}^2$.
   Similarly to the proof of Lemma~\ref{lem:square_bound}, we split $E$ into two subsets
   \[
     E_0 := E \cap (V \cdot \widetilde{V}) \quad\text{ and }\quad E_1 := E \setminus E_0.
   \]
   For every $1 \leqslant i \leqslant r$, the cardinality of $E \cap (z_{\ell + i} \cdot V)$ does not exceed the multiplicity of $z_{\ell + i}$ in the multiset $D$ from~\ref{step:setD_C4}.
   Therefore, $|E_0| \leqslant \sum\limits_{i = 1}^r \operatorname{mult}[i]$.
   
   Consider a pseudograph $G$ with $r$ vertices numbered from $1$ to $r$ corresponding to $z_{\ell + 1}, \ldots, z_{\ell + r}$, respectively.
   For every element $m \in E_1$, we fix a representation $m = z_{\ell + i}z_{\ell + j}$, and add an edge connecting vertices $i$ and $j$ in $G$ (this will be a loop of $i = j$).
   We claim that pseudograph $G$ will be $C4^\ast$-free. 
   Indeed, if there is a cycle formed by edges $m_1, m_2, m_3, m_4 \in E_0$, then we will have $m_1 \cdot m_3 = m_2 \cdot m_4$.
   Moreover, $\{m_1, m_3\} \cap \{m_2, m_4\} = \varnothing$, so such a relation contradicts the condition on $E$ imposed by~\ref{step:constructE}.
   Finally, a monomial $m \in E$ can correspond to a loop in $G$ only if it is a square, that is, all the degrees in $m$ are even.
   Hence $E_1$, the total number of edges in $G$, does not exceed $C(r, c)$
   
   In total, we have
   \[
     |E| = |E_0| + |E_1| \leqslant \sum\limits_{i = 1}^r \operatorname{mult}[i] + C(r, c),
   \]
   so $r$ satisfies~\eqref{eq:C4_bound}.
   The minimality of $k$ implies that $r \geqslant k$.
   Thus, $r + \ell \geqslant N$, so $z_1, \ldots, z_\ell$ cannot be extended to a quadratization of order less than $N$.
\end{proof}

\begin{remark}[Cycles of even length]
   One can modify this rule to use graphs not containing cycles of even length. 
   In this case, the set $E$ from~\ref{step:constructE} of Algorithm~\ref{alg:early_C4} would satisfy the condition that there are no multi-subsets of equal cardinality and with equal product.
   However, this approach did not work that well in practice, in particular, due to the overhead for finding such $E$.
\end{remark}

\subsection{Performance of the pruning rules}

Table~\ref{table:main_bench} below shows the performance of our algorithm with a different combination of the pruning rules employed.
It shows that the rules substantially speed up the computation and that Algorithm~\ref{alg:early_C4} is particularly successful in higher dimensions.

\begin{table}[H]
\centering
\setlength{\tabcolsep}{0.5em} 
{\renewcommand{\arraystretch}{1.2}
\begin{tabular}{|l|c|c|c|c|c|}
\hline
ODE system &  Dimension & No pruning & Alg.~\ref{alg:early_quadratic} & Alg.~\ref{alg:early_C4} & Alg.~\ref{alg:early_quadratic} \& \ref{alg:early_C4}\\
\hline
Circular(8) & $2$ & $4293 \pm 445$  & $497 \pm 5$ & $526 \pm 8$ & $453 \pm 7$\\
Hill(20) & $3$ & $3.4 \pm 0.1$  & $3.0 \pm 0.1$ & $2.4 \pm 0.1$ & $2.4 \pm 0.1$\\
Hard(2) & $3$ & $106.3 \pm 1.0$  & $19.6 \pm 1.1$ & $20.1 \pm 0.6$ & $16.7 \pm 0.6$\\
Hard(4) & $3$ & $360.1 \pm 5.6$  & $107.5 \pm 2.4$ & $108.8 \pm 2.1$ & $96.6 \pm 1.5$\\
Monom(3) & $3$ & $552.9 \pm 10.9$  & $85.7 \pm 4.2$ & $124.7 \pm 5.5$ & $84.2 \pm 3.3$\\
Cubic Cycle(6) & $6$ & $187.3 \pm 0.8$  & $43.6 \pm 0.6$ & $20.0 \pm 0.5$ & $20.1 \pm 0.3$\\
Cubic Cycle(7) & $7$ & $2002 \pm 6.4$ & $360.7 \pm 1.1$ & $150.2 \pm 1.3$ & $160.9 \pm 5.9$ \\
Cubic Bicycle(7) & $7$ & $1742 \pm 89$ & $73.2 \pm 0.6$ & $29.8 \pm 0.3$ & $30.5 \pm 0.2$ \\
Cubic Bicycle(8) & $8$ & $4440+$ & $175.4 \pm 4.0$ & $64.8 \pm 0.5$ & $68.9 \pm 0.7$ \\
\hline
\end{tabular}
}
\vspace{1mm}
\caption{Comparison of the pruning rules used by our algorithm. Values in the cells represent an average time with the standard deviation in seconds.}\label{table:main_bench}
\end{table}


\section{Performance and Results}\label{sec:performace}

We have implemented our algorithm in Python, and the implementation is available at~\url{https://github.com/AndreyBychkov/QBee/tree/0.5.0}.
We compare our algorithm with the one proposed in~\cite{CRN}.
For the comparison, we use the set of benchmarks from~\cite{CRN} and add a couple of new ones (described in the Appendix).

The results of the comparison are collected in Table~\ref{table:CRM_comp}.
All computation times are given either in milliseconds or in seconds and were obtained on a laptop with the following parameters: Intel(R) Core(TM) i7-8750H CPU @ 2.20GHz, WSL Windows 10 Ubuntu 20.04, CPython 3.8.5. 
From the table, we see that the only cases when the algorithm from~\cite{CRN} runs faster are when it does not produce an optimal quadratization (while we do). 
Also, cases, when the algorithm from~\cite{CRN} is not able to terminate, are marked as "—" symbol.

\begin{table}[H]
\centering
\setlength{\tabcolsep}{0.5em} 
{\renewcommand{\arraystretch}{1.2}
\begin{tabular}{|l|c|c|c|c|}
\hline
ODE system &  Biocham time & Biocham order & Our time & Our order\\
\hline
Circular(3), \textbf{ms} & $83.2 \pm 0.1$ & 3 & $5.1 \pm 0.1$ & 3\\
Circular(4), \textbf{ms} & $106.7 \pm 2.3$ & 4 & $164.8 \pm 32.3$ & 4\\
Circular(5), \textbf{ms} & $596.2 \pm 10.9$ & 4 & $20.0 \pm 0.1$ & 4\\
Circular(6), \textbf{s} & $37.6 \pm 0.4$ & 5 & $4.2 \pm 0.1$ & 5\\
Circular(8), \textbf{s} & \textbf{—} & \textbf{—} & $453.3 \pm 6.9$ & 6 \\
Hard(3), \textbf{s} & $1.09 \pm 0.01$ & \color{red}11 & $8.6 \pm 0.2$ & 9\\
Hard(4), \textbf{s} & $20.2 \pm 0.3$ & \color{red}13 & $96.9 \pm 1.5$ & 10\\
Hill(5), \textbf{ms} & $87.8 \pm 0.9$ & 2 & $4.6 \pm 0.0$ & 2\\
Hill(10), \textbf{ms} & $409.8 \pm 5.6$ & 4 & $49.7 \pm 1.3$ & 4\\
Hill(15), \textbf{s} & $64.1 \pm 0.4$ & 5 & $0.34 \pm 0.1$ & 5\\
Hill(20),\textbf{s} & \textbf{—} & \textbf{—} & $2.4 \pm 0.1$ & 6 \\
Monom(2), \textbf{ms} & $96.4 \pm 1.6$ & \color{red}4 & $15 \pm 0.1$ & 3\\
Monom(3), \textbf{s} & $0.44 \pm 0$ & \color{red}13 & $84.2 \pm 3.3$ & 10\\
Cubic Cycle(6), \textbf{s} & \textbf{—} & \textbf{—} & $20.1 \pm 0.3$ & 12\\
Cubic Cycle(7), \textbf{s} & \textbf{—} & \textbf{—} & $160.9 \pm 5.9$ & 14\\
Cubic Bicycle(7), \textbf{s} & \textbf{—} & \textbf{—} & $30.5 \pm 0.2$ & 14\\
Cubic Bicycle(8), \textbf{s} & \textbf{—} & \textbf{—} & $68.9 \pm 0.7$ & 16\\
\hline
\end{tabular}
}
\vspace{1mm}
\caption{Comparison of our implementation with the algorithm \cite{CRN} on a set benchmarks}\label{table:CRM_comp}.
\end{table}

\section{Remarks on the complexity}\label{sec:complexity}

It has been conjectured in~\cite[Conjecture~1]{CRN} that the size of an optimal monomial quadratization may be exponential in the number of monomials of the input system in the worst case.
Interestingly, this is not the case if one allows monomials with negative powers (i.e., Laurent monomials): Proposition~\ref{prop:negative} shows that there exists a quadratization with the number of new variables being linear in the number of monomials in the system.

\begin{proposition}\label{prop:negative}
  Let $\bar{x}' = \bar{f}(\bar{x})$, where $\bar{x} = (x_1, \ldots,x_n)$, be a system of ODEs with polynomial right hand sides.
  For every $1 \leqslant i \leqslant n$, we denote the monomials in the right-hand side of the $i$-th equation by $m_{i, 1}, \ldots, m_{i, k_i}$.
  Then the following set of new variables (given by Laurent monomials) is a quadratization of the original system:
  \[
  z_{i, j} := \frac{m_{i, j}}{x_i} \text{ for every } 1 \leqslant i \leqslant n, \; 1\leqslant j \leqslant k_i.
  \]
\end{proposition}

\begin{proof}
  Since $m_{i, j} = z_{i, j} x_i$, the original equations can be written as quadratic in the new variables.
  Let the coefficient in the original system in front of $m_{i, j}$ be denoted by $c_{i, j}$.
  We consider any $1 \leqslant i \leqslant n$, $1 \leqslant j \leqslant k_j$:
  \[
    z_{i, j}' = \sum\limits_{s = 1}^n f_s(\mathbf{x}) \frac{\partial z_{i, j}}{\partial x_s} = \sum\limits_{s = 1}^n \sum\limits_{r = 1}^{k_s} c_{s, r} m_{s, r} \frac{\partial z_{i, j}}{\partial x_s}.
  \]
  Since $\frac{\partial z_{i, j}}{\partial x_s}$ is proportional to $\frac{z_{i, j}}{x_s}$, the monomial $m_{s, r} \frac{\partial z_{i, j}}{\partial x_s}$ is proportional to a quadratic monomial $z_{s, r}z_{i, j}$, so we are done.
\end{proof}

\begin{remark}[{Relation to the $[2]$-sumset cover problem}]\label{rem:2sumsetcover}
   The [2]-sumset cover problem~\cite{SSC} is, given a finite set $S \subset \mathbb{Z}_{> 0}$ of positive integers, find a smallest set $X \subset \mathbb{Z}_{> 0}$ such that $S \subset X \cup \{x_i + x_j \mid x_i, x_j \in X\}$.
   It has been shown in~\cite[Proposition~1]{Fagnot2009} that the [2]-sumset cover problem is APX-hard, moreover, the set $S$ used in the proof contains $1$.
   We will show how to encode this problem into the optimal monomial quadratization problem thus showing that the latter is also APX-hard (in the number of monomials, but not necessarily in the size of the input).
   For $S = \{s_1, \ldots, s_n\} \subset \mathbb{Z}_{> 0}$ with $s_1 = 1$, we define a system
   \[
     x_1' = 0,\quad x_2' = \sum\limits_{i = 1}^n x_1^{s_i}.
   \]
   Then a set $X = \{1, a_1, \ldots, a_\ell\}$ is a minimal [2]-sumset cover of $S$ iff $x_1^{a_1}, \ldots, x_1^{a_\ell}$ is an optimal monomial quadratization of the system.
\end{remark}


\section{Conclusions and Open problems}\label{sec:conclusion}

In this paper, we have presented the first practical algorithm for finding an optimal monomial quadratization.
Our implementation compares favorably with the existing software and allows us to find better quadratizations for already used benchmark problems.
We were able to compute quadratization for ODE systems which could not be tackled before. 

We would like to mention several interesting open problems:
\begin{enumerate}
    \item Is it possible to describe a finite set of monomials which must contain an optimal quadratization? This would allow using SAT-solving techniques of~\cite{CRN} as described in Section~\ref{sec:prior}. 
    \item As has been shown in~\cite{Foyez}, general polynomial quadratization may be of a smaller dimension than an optimal monomial quadratization. 
    This poses a challenge: design an algorithm for finding optimal polynomial quadratization (or at least a smaller one than an optimal monomial).
    \item How to search for optimal monomial quadratizations if negative powers are allowed (see Section~\ref{sec:complexity})?
    \item How to design a faster algorithm for approximate quadratization (that is, finding a quadratization which is close to the optimal) with guarantees on the quality of the approximation?
\end{enumerate}

\bibliographystyle{splncs04}
\bibliography{bibliography}

\section*{Appendix: Benchmark systems}

Most of the benchmark systems used in this paper (in Tables~\ref{table:CRM_comp} and~\ref{table:CRM_comp}) are described in~\cite{CRN}.
Here we show additional benchmarks we have introduced:
\begin{enumerate}
    \item \emph{Cubic Cycle($n$).} For every integer $n > 1$, we define a system in variables $x_1, \ldots, x_n$ by
    \[
      x_1' = x_2^3,\; x_2' = x_3^3,\;\ldots,\; x_n' = x_1^3.
    \]
    
    \item \emph{Cubic Bicycle($n$).}
    For every integer $n > 1$, we define a system in variables $x_1, \ldots, x_n$ by
    \[
      x_1' = x_n^3 + x_2^3,\; x_2' = x_1^3 + x_3^3,\;\ldots,\; x_n' = x_{n - 1}^3 + x_1^3.
    \]
\end{enumerate}

\end{document}